\documentclass[11pt]{amsart}
\usepackage[leqno]{amsmath}
\usepackage{amsfonts, amssymb, amsopn, color} 
\usepackage[utf8]{inputenc}

\usepackage{tikz}
\usetikzlibrary{shapes.geometric}
\usetikzlibrary{calc}

\newcommand{\point}[1]{\node[circle,inner sep = 0pt,minimum size =0pt]}
\newcommand{\vertex}[1]{\node[circle,inner sep = 0pt,minimum size =3pt,fill = #1]}
\newcommand{\Vertex}[1]{\node[circle,inner sep = 0pt,minimum size =5pt,fill = #1]}

\usepackage{graphicx}
\usepackage{xspace}

\usepackage{epsfig}
\usepackage{color}
\usepackage{amsmath}

\newtheorem{lemma}{Lemma}
\newtheorem{proposition}{Proposition}
\newtheorem{theorem}{Theorem}
\newtheorem{corollary}{Corollary}

\theoremstyle{definition}

\usepackage{color}
\newcommand{\commentout}[1]{}

\DeclareMathOperator{\conv}{conv}
\DeclareMathOperator{\id}{id}

\begin{document}

\thispagestyle{empty}
\centerline{\Large\bf Isometric embedding of Busemann surfaces into $L_1$}

\vspace{10mm}

\centerline{{\sc Jérémie Chalopin}, {\sc Victor Chepoi}, and {\sc Guyslain Naves}}

\vspace{3mm}

\date{\today}

\medskip
\begin{small}
\centerline{Laboratoire d'Informatique Fondamentale, Aix-Marseille Universit\'e and CNRS,}
\centerline{Facult\'e des Sciences de Luminy, F-13288 Marseille Cedex 9, France}

\centerline{\texttt{\{jeremie.chalopin, victor.chepoi, guyslain.naves\}@lif.univ-mrs.fr}}
\end{small}

\bigskip\bigskip\noindent {\footnotesize {\bf Abstract.}  In this
paper, we prove that any non-positively curved 2-dimensional surface
(alias, Busemann surface) is isometrically embeddable into $L_1$. As a
corollary, we obtain that all planar graphs which are 1-skeletons of
planar non-positively curved complexes with regular Euclidean polygons
as cells are $L_1$-embeddable with distortion at most $2+\pi/2<4$. Our
results significantly improve and simplify the results of the recent
paper {\it A. Sidiropoulos, Non-positive curvature, and the planar
embedding conjecture, FOCS 2013.}}

\section{Avant-propos}
Isometric and low distortion embeddings of finite and infinite metric
spaces into $L_p$-spaces is one of the main subjects in the theory of
metric spaces. Work in this area was initiated by Cayley in 1841 and
continued in the first half of the 20th century by Fréchet, Menger,
Schoenberg, and Blumenthal. Since these days it is known that all
metric spaces isometrically embed into $L_{\infty}$. Metric spaces
isometrically embeddable into $L_2$ were characterized by Menger and
Schoenberg. Even if  embeddability into $L_1$ can be defined in several
equivalent ways and a few necessary conditions for $L_1$-embedding
are known, metric spaces isometrically embeddable into $L_1$ cannot be
characterized in an efficient way, because deciding whether a finite
metric space is $L_1$-embeddable is NP-complete.  On the other hand,
many classes of metric spaces (Euclidean and spherical metrics, tree
and outerplanar metrics, as well as graph metrics of some classes of
graphs) are known to be $L_1$-embeddable; for a full account of the
theory of isometric embeddings, see the book \cite{DeLa}.

Although already simple metric spaces are not $L_1$-embeddable,
Bourgain \cite{Bou} established that any metric space on $n$ points
can be embedded into $L_1$ with $O(\log n)$ (multiplicative)
distortion and this important result has found numerous algorithmic
applications (for a theory of low distortion embeddings of metric
spaces and its algorithmic applications, the interested reader can
consult the book \cite{Mat} and the survey \cite{InMa}). One of main
open problems in this domain is the so-called {\it planar embedding
conjecture} asserting that all planar metrics (i.e., metrics of planar
graphs) can be embedded into $L_1$ with constant distortion. This
conjecture was established for series-parallel graphs \cite{GuNeRaSi};
on the other hand, several classes of planar graphs are known to be
$L_1$-embeddable (see the book \cite{DeLa} and the survey
\cite{BaCh_survey}), in particular, it was shown in \cite{ChDrVa_jalg}
that the three basic classes of non-positively curved planar graphs
(so-called, (3,6),(6,3), and (4,4)-graphs) are $L_1$-embeddable.

Recently, Sidiropoulos \cite{Si} proved that for any finite set $Q$ of
a non-positively curved planar surface $(S,d),$ the metric space
$(Q,d)$ is $L_1$-embeddable with constant distortion. As a
consequence, all planar graphs which give rise to non-positively
curved surfaces can be embedded into $L_1$ with constant
distortion.\footnote{\noindent Non-positively curved metric spaces
constitute a large class of geodesic metric spaces at the heart of
modern metric geometry and playing an essential role in geometric
group theory \cite{BrHa,Pa}.} The proof-method in \cite{Si} uses at
minimum the geometry of  $(S,d)$ and essentially
employs probabilistic techniques.  Using the convexity of the distance
function, Sidiropoulos ``approximates'' $(Q,d)$ by special planar
graphs called ``bundles", which he shows to be $L_1$-embeddable with
constant distortion. To provide such an embedding, the author searches
for a good distribution over a special type of cuts (bipartitions of
$Q$), called ``monotone cuts", defined on bundles.  Searching for this
good distribution is the most technically involved part of the paper
\cite{Si}.

In this paper, we prove that in fact all non-positively curved planar
surfaces $(S,d)$ (which we call {\it Busemann surfaces}) are
$L_1$-embeddable (without any distortion). This significantly improves
the result of \cite{Si}.  Our approach is geometric and
combinatorial. First, we establish some elementary properties of
convexity in Busemann surfaces, analogous to the properties
of usual convexity in ${\mathbb R}^2$.  Then we use some of these
properties to show that for any finite set $Q$ in general position of
$S$, $(Q,d)$ is $L_1$-embeddable. For this, we extend to Busemann
surfaces the proof-method of a combinatorial Crofton lemma given by
R. Alexander \cite{Al} for finite point-sets in general position in
${\mathbb R}^2$ endowed with a metric in which lines are
geodesics. Using local perturbations of points, we extend our result to
all finite sets $Q$ of $S$. Then the fact that $(S,d)$ is
$L_1$-embeddable follows from a compactness result of \cite{BrDhCaKr}
about $L_p$-embeddings.

\section{Preliminaries}
\subsection{$L_1$-embeddings}

A metric space $(X,d)$ is {\it isometrically embeddable} into a metric
space $(X',d')$ if there exists a map $\varphi: X\rightarrow X'$ such
that $d'(\varphi(x),\varphi(y))=d(x,y)$ for any $x,y\in X$.  More
generally, $\varphi: X\mapsto X'$ is an {\it embedding with
  (multiplicative) distortion} $c \ge 1$ if $d(x,y) \le
d'(\varphi(x),\varphi(y)) \le c \cdot d(x,y)$ for all $x,y \in X$
(non-contractive embedding), or if $\frac{1}{c} \cdot d(x,y) \le
d'(\varphi(x),\varphi(y)) \le d(x,y)$ for all $x,y \in X$
(non-expansive embedding). Let $(\Omega,{\mathcal A},\mu)$ be a measure space
consisting of a set $\Omega$, a $\sigma$-algebra $\mathcal A$ of
subsets of $\Omega$, and a measure $\mu$ on $\mathcal A$. Given a
function $f:\Omega\rightarrow {\mathbb R},$ its $L_1$-{\it norm} is
defined by $\|f\|_1=\int_{\Omega} \vert f(w)\vert \mu (dw).$ Then
$L_1(\Omega,{\mathcal A},\mu)$ denotes the set of functions
$f:\Omega\rightarrow {\mathbb R}$ which satisfy $\|f\|_1<\infty$. The
$L_1$-norm defines a metric on $L_1(\Omega,{\mathcal A},\mu)$ by
taking $\| f-g\| _1$ as the distance between two functions $f,g\in
L_1(\Omega,{\mathcal A},\mu)$.  A metric space $(X,d)$ is said to be
$L_1$--{\it embeddable} if there exists an isometric embedding of
$(X,d)$ into $L_1(\Omega,{\mathcal A},\mu)$ for some measure space
$(\Omega,{\mathcal A}, \mu)$ \cite{DeLa}. If $\Omega$ is finite (say,
$|\Omega|=n$) and ${\mathcal A}=2^{\Omega},$ the resulting space
$L_1(\Omega,{\mathcal A},\mu)$ is the $n$-dimensional $l_1$-space
$({\mathbb R}^n,d_1),$ where the $l_1$-distance between two points
$x=(x_1,\ldots,x_n)$ and $y=(y_1,\ldots,y_n)$ is
$d_1(x,y)=\sum_{i=1}^n |x_i-y_i|.$ A metric space $(X_n,d)$ on $n$
points is $l_1$-{\it embeddable} (and $d$ is called an $l_1$-{\it
  metric}) if there exists an isometric embedding of $(X_n,d)$ into
some $l_1$-space $({\mathbb R}^m,d_1)$. It is well known \cite[Chapter
  4]{DeLa} that the set of all $l_1$-metrics on $X_n$ forms a closed
cone CUT$_n$ in ${\mathbb R}^{\frac{n(n-1)}{2}}$, called the {\it cut
  cone}. CUT$_n$ is generated by the cut semimetrics $\delta_S$ for
$S\subseteq X_n,$ where $\delta_S(x,y)=1$ if $|S\cap \{ x,y\}|=1$ and
$\delta_S(x,y)=0$ otherwise. A well--known compactness result of
\cite{BrDhCaKr} implies that $L_1$--embeddability of a metric space is
equivalent to $l_1$--embeddability of its finite subspaces.

\subsection{Geodesics and geodesic metric spaces}

In this subsection, we recall some definitions and notations on
geodesic metric spaces; we closely follow the books
\cite{BrHa} and \cite{Pa}.  Let $(X,d)$ be a metric space.  A {\it
path} in $X$ is a continuous map $\gamma:[a,b]\rightarrow X$, where
$a$ and $b$ are two real numbers with $a\le b$.  If $\gamma(a)=x$ and
$\gamma(b)=y,$ then $x$ and $y$ are the {\it endpoints} of
$\gamma$ and that $\gamma$ joins $x$ and $y$.  A {\it geodesic path}
(or simply a {\it geodesic}) in $X$ is a path $\gamma:
[a,b]\rightarrow X$ that is distance-preserving, that is, such that
$d(\gamma(s),\gamma(t))=|s-t|$ for all $s,t\in [a,b].$ A {\it geodesic
line} (or simply a {\it line}) is a distance-preserving map
$\gamma:{\mathbb R}\rightarrow X$ and a {\it geodesic ray} (or simply
a {\it ray}) is a distance-preserving map $\gamma:
[0,\infty)\rightarrow X.$ A path $\gamma: [a,b]\rightarrow X$ is said
to be a {\it local geodesic} if for all $t$ in $(a,b)$ one can find a
closed interval $I(t) \subseteq [a,b]$ containing $t$ in its interior
such that the
restriction of $\gamma$ on $I(t)$ is geodesic.  A metric
space $X$ is {\it geodesic} if every pair of points in $X$ can be
joined by a geodesic. A {\it uniquely geodesic space} is a geodesic
space in which every pair of points can be joined by a unique
geodesic.

\subsection{Non-positively curved spaces}

We continue with the definitions of non-positively curved spaces in
the sense of Alexandrov \cite{BrHa} and of Busemann \cite{Pa}.  A {\it
geodesic triangle} $\Delta:=\Delta (x_1,x_2,x_3)$ in a geodesic metric
space $(X,d)$ consists of three points in $X$ (the vertices of
$\Delta$) and a geodesic between each pair of vertices (the edges of
$\Delta$). A {\it comparison triangle} for $\Delta (x_1,x_2,x_3)$ is a
triangle $\Delta (x'_1,x'_2,x'_3)$ in the Euclidean plane ${\mathbb
E}^2$ such that $d_{{\mathbb E}^2}(x'_i,x'_j)=d(x_i,x_j)$ for $i,j\in
\{ 1,2,3\}.$ A geodesic metric space $(X,d)$ is a {\it
CAT(0) space} (or a {\it non-positively curved space in the sense of
Alexandrov}) \cite{Gr} if for all geodesic triangles $\Delta
(x_1,x_2,x_3)$ of $X$, if $y$ is a point on the side of $\Delta(x_1,x_2,x_3)$ with
vertices $x_1$ and $x_2$ and $y'$ is the unique point on the line
segment $[x'_1,x'_2]$ of the comparison triangle
$\Delta(x'_1,x'_2,x'_3)$ such that $d_{{\mathbb E}^2}(x'_i,y')=
d(x_i,y)$ for $i=1,2,$ then $d(x_3,y)\le d_{{\mathbb
E}^2}(x'_3,y').$ CAT(0) spaces have many fundamental properties and can be
characterized in several natural ways (for a full account of
this theory consult the book \cite{BrHa}).

A {\it Busemann space} (or a {\it non-positively curved space in the
sense of Busemann}) is a geodesic metric space $(X,d)$ in which the
distance function between any two geodesics is convex: for any two
reparametrized geodesics $\gamma: [a,b]\rightarrow X$ and $\gamma':
[a',b']\rightarrow X$ the map $f_{\gamma,\gamma'}(t):[0,1]\rightarrow
{\mathbb R}$ defined by
$f_{\gamma,\gamma'}(t)=d(\gamma((1-t)a+tb),\gamma'((1-t)a'+tb'))$ is
convex.  Each CAT(0) space is a Busemann space, but not
vice-versa. However, Busemann spaces still satisfy most of fundamental
properties of CAT(0) spaces: they are contractible, uniquely geodesic,
local geodesics are geodesics, and geodesics vary continuously with
their endpoints (for these and other results on Busemann spaces
consult the book \cite{Pa}).  Busemann spaces and CAT(0) spaces are
the same in the case of smooth Riemannian manifolds and of piecewise
Euclidean or hyperbolic complexes (because in the latter case, the
CAT(0) property is equivalent to the uniqueness of geodesics; see
\cite[Theorem 5.4]{BrHa}).

\subsection{Busemann surfaces}

A {\it planar surface} $S$ is a 2-dimensional manifold without
boundary, i.e., $S$ is homeomorphic to  the plane ${\mathbb R}^2$.  A
geodesic metric space $(S,d)$ is called a {\it Busemann surface}
if $S$ is a planar surface and the metric space
$(S,d)$ is a Busemann space.  Notice that
each point $x$ of $S$ has an open neighborhood
$B(x,\epsilon)$ which is homeomorphic to an open ball in the plane.

Particular instances of Busemann surfaces are
non-positively curved piecewise-Euclidean (PE) (or piecewise
hyperbolic) planar complexes without boundary. In fact, as we will
show below, any finite non-positively curved planar complex can be
extended to a Busemann surface. Recall that a {\it
planar PE complex} $X$ is obtained from a (not necessarily finite)
planar graph $G$ by replacing each inner face of $G$ having $n$ sides
by a convex $n$-gon in the Euclidean plane. Then $X$ is called a {\it
regular planar complex} if each face of $G$ with $n$ sides is replaced
by a regular $n$-gon in the plane. Note that the graph $G$ is the
1-skeleton of $X$.  The complex $X$ is called a {\it non-positively
curved planar complex} if the sum of angles around each inner vertex
of $G$ is at least $2\pi$ or, equivalently, if $X$ endowed with the
intrinsic $l_2$-metric is  uniquely geodesic. We will call a planar graph $G$
a {\it Busemann graph} (or a {\it non-positively curved planar graph})
if $G$ is the 1-skeleton of a regular non-positively curved planar complex.  Basic
examples of Busemann graphs are so-called
(3,6),(4,4), and (6,3)-graphs (a planar graph $G$ embedded into the
plane is called a $(p,q)$-{\it graph} if the degrees of all inner
vertices are at least $p$ and all inner faces have lengths at least
$q$). It was shown in \cite{ChDrVa_jalg} (see also \cite[Proposition
8.6]{BaCh_survey}) that all (3,6),(4,4), and (6,3)-graphs are
$l_1$-embeddable (for other properties of these graphs, see \cite{BaPe}).
In this paper, we will present a Busemann
planar graph which is not $l_1$-embeddable.

To embed a finite non-positively curved planar complex
$X$ into a Busemann surface $S$, to each boundary edge
$e$ of $X$ we add a closed halfplane $H_e$ of ${\mathbb R}^2$ so that
$e$ is a segment of the boundary of $H_e$. If two boundary edges
$e,e'$ of $X$ share a common endvertex $x,$ then $H_e$ and $H_e'$ will be
glued along the rays of their boundaries emanating from $x$ which are
disjoint from $e$ and $e'$. 
It can be easily seen that the resulting planar surface $S$ is CAT(0) and
that $X$ isometrically embeds into $S$.

\subsection{Main results}

We continue with the formulation of the main results of this paper:

\begin{theorem} \label{main} If $(S,d)$ is a Busemann surface, then $(S,d)$ is $L_1$-embeddable.
\end{theorem}

\begin{corollary} \label{graph} Any Busemann graph $G$ endowed with its standard-graph metric $d_G$
admits a non-expansive  $L_1$-embedding with distortion at most $2+\pi/2$.
\end{corollary}

To prove Theorem \ref{main}, by a compactness result of
\cite{BrDhCaKr}, it suffices to show
that for any finite subset $Q$ of $S,$ $(Q,d)$ is
$l_1$-embeddable. For this, first we show that any finite subset of
$S$ in general position (no three points on a common geodesic line) is
$l_1$-embeddable.  Then, using this result and a local perturbations
of points, we prove that there exists a sequence $d_i$ of $l_1$-metrics
on $Q$ converging to $d$. Since the cone  of $l_1$-metrics on
$Q$  is closed \cite{DeLa}, $(Q,d)$ is $l_1$-embeddable as well.  Our proof
that any finite set of $S$ in general
position is $l_1$-embeddable is based on a beautiful Crofton formula
by Alexander \cite{Al} established for finite point-sets in general position of
${\mathbb R}^2$ endowed with a metric in which lines are
geodesics (for another use of this formula, see \cite{ChFi}). To generalize
this result to Busemann surfaces
$S$, we extend to $S$ some elementary properties  of the usual convexity
in the plane (see also \cite{Ma} for some similar properties  of planar
 CAT(0) complexes and  \cite{DhGoHoPoSm}  for such properties for topological affine planes).
 This is done in Section \ref{sec:npcsurfaces}. The proof
of the Crofton formula for Busemann surfaces is given in
Section \ref{sec:crofton} and the proof of the main results is
completed in Section \ref{sec:main}.

\section{Geodesic lines and convexity}\label{sec:npcsurfaces}

In this section, we present elementary properties of geodesic lines
and convex sets in Busemann planar surfaces $(S,d)$. For
two points $x,y\in S,$ we  denote by $[x,y]$ the unique geodesic segment
joining $x$ and $y$.
A set of points $Q$ of $S$ is in {\it general position} if no three points
of $Q$ are collinear, i.e., lie on a common line of $S$. For three points $x,y,z$ of $S$, we will
denote by $\Delta^*(x,y,z)$  the closed region
of $S$ bounded by the geodesics $[x,y],[y,z],[z,x]$ of the geodesic triangle
$\Delta(x,y,z).$ A set $R\subseteq S$ is called {\it convex} if
$[p,q]\subseteq R$ for any $p,q\in R$. For a set $Q$ of $S$ we denote
by conv$(Q)$ the smallest convex set containing $Q$ and call conv$(Q)$
the {\it convex hull} of $Q$.

\subsection{Geodesic lines} \label{sec:lines}

A geodesic metric space $(X,d)$ is said to have the {\it geodesic
  extension property} if for every local geodesic $\gamma:
[a,b]\rightarrow X,$ with $a<b,$ there exists $\epsilon>0$ and a local
geodesic $\gamma': [a,b+\epsilon]\rightarrow X$ such that
$\gamma'|_{[a,b]}=\gamma$. In Busemann spaces local
geodesics are geodesics, hence for such spaces the geodesic extension
property is equivalent to the fact that the geodesic between any two
distinct points can be extended to a geodesic line. It was established
in Proposition 5.12 of \cite{BrHa} that any CAT(0) space that is
homeomorphic to a finite dimensional manifold has the geodesic
extension property. By \cite[Footnote 24]{BrHa} and the proof of this
proposition, an analogous statement holds for Busemann
spaces. Therefore, for Busemann surfaces $(S,d)$
we obtain:

\begin{lemma} \label{extension}  $S$ has the geodesic extension property.
\end{lemma}

The next lemma immediately follows from the definition of Busemann spaces.

\begin{lemma} \label{ball-convexity}
Closed balls of $S$  are convex.
\end{lemma}

\begin{lemma} \label{halfplanes}  Any geodesic line $\ell$ partitions $S$ into two connected components.
\end{lemma}

\begin{proof}
Let $x$ be a point of $\ell$ and let $B(x,\epsilon)$ be a closed ball
centered at $x$. Since $B(x,\epsilon)$ is convex and $S$ is a planar
surface, $B(x,\epsilon)\cap \ell=[p,q]$ and the segment $[p,q]$
partitions $B(x,\epsilon)$ in two connected components. Let $u,v$ be
two points from different components of $B(x,\epsilon/2)\setminus
[p,q]$. By Lemma~\ref{ball-convexity}, $p \notin [u,v]$. Since
$B(x,\epsilon)$ is convex, necessarily $[u,v]$ intersects $[p,q]$; let
$x'\in [u,v]\cap [p,q]$.  Now, if $\ell$ does not separate $S$, then
$u$ and $v$ can be connected by a path $\gamma$ not intersecting
$\ell$.  Let $R$ be the region of $S$ bounded by $\gamma$ and
$[u,v]$. Then one of the points $p,q,$ say $p$, belongs to $R$.  The
ray $r_{x'}$ of $\ell$ emanating from $x'$ and passing via $p$ enters
the region $R$, hence $r_{x'}$ must intersect the boundary of $R$ in a
point $w$ different from $x'$. Then $p\in [x',w]$.  Since $\ell\cap
\gamma=\varnothing,$ necessarily $w \in [u,v]$ and we conclude that $x'$
and $w$ are joined by two different geodesics, one is a portion of
$\ell$ passing via $p \notin [u,v]$ and the second is a
portion of $[u,v],$ a contradiction.
\end{proof}


For a line $\ell$, we denote by $\mathring{H}^-_{\ell}$
and $\mathring{H}^+_{\ell}$ the two connected components
of $S\setminus \ell$, respectively, and call them {\it open
halfplanes}. The {\it closed halfplanes} defined by $\ell$ are the
sets $H^-_{\ell}=\mathring{H}^-_{\ell}\cup \ell$ and
$H^+_{\ell}=\mathring{H}^+_{\ell}\cup \ell$. Since each
line $\ell$ is convex, the following result is straightforward.

\begin{lemma} \label{closed_halfplane}  The closed halfplanes $H^-_{\ell}$ and $H^+_{\ell}$  are
convex sets of $S$.
\end{lemma}

\begin{lemma} \label{two-lines}
Let $\ell$ and $\ell'$ be two intersecting geodesic lines such that
$\ell'$ is contained in the closed halfplane $H_{\ell}^+$ defined by
$\ell$, $x \in \ell \cap \ell'$, and let $r_1,\ldots r_4$ be the four
rays emanating from $x$ defined as in Figure~\ref{fig:two-lines}. Then
$r_1 \cup r_3$ and $r_2 \cup r_4$ are also geodesic lines.
\end{lemma}

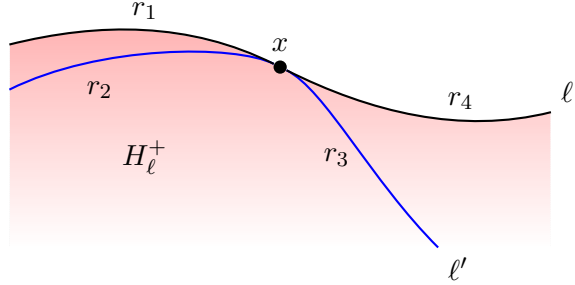
\begin{figure}
\begin{center}
\begin{tikzpicture}[x=0.3cm,y=0.3cm,>=latex]
\fill[shading=axis,top color = red!30,bottom color = white,draw=white]
(1,17) .. controls (5,18) and (9,18) .. (13,16)
         .. controls (17,14) and (21,13) .. (25,14)
         -- (25,8) -- (1,8) -- (1,17);
\Vertex{black} (nx13y16) at (13,16) {};
\draw[blue,thick] (nx13y16) .. controls (15,15) and (17,11) .. (20,8);
\draw[blue,thick] (1,15) .. controls (5,17) and (11,17) .. (nx13y16);
\draw[black,thick] (nx13y16) .. controls (17,14) and (21,13) .. (25,14);
\draw[black,thick] (1,17) .. controls (5,18) and (9,18) .. (nx13y16);
\draw (15.5,12) node {$r_3$};
\draw (5,15) node {$r_2$};
\draw (21,14.5) node {$r_4$};
\draw (7,18.5) node {$r_1$};
\draw (20,8) node[anchor = north west] {$\ell'$};
\draw (25,14) node[anchor = south west] {$\ell$};
\draw (7,12) node {$H^+_\ell$};
\draw (13,17) node {$x$};
\end{tikzpicture}
\end{center}
\caption{Illustration to Lemma~\ref{two-lines}.} 
\label{fig:two-lines}
\end{figure}

\begin{proof}
We will prove that $\ell_0=r_2 \cup r_4$ is a geodesic line. If $\ell
\cap \ell'= [x',x'']$ with $x' \ne x''$, then $\ell_0$ is a local
geodesic because it is covered by the rays of $\ell$ and $\ell'$
emanating from $x'$ and $x''$ sharing the geodesic segment
$[x',x'']$. Thus $\ell_0$ is a geodesic line.

Now, suppose that $\ell \cap \ell'=\{ x\}$. Pick any two points $y \in
r_2$ and $z \in r_4$ and suppose by way of contradiction that $[y,z]$
is not contained in $\ell_0$. Moreover, we can suppose without loss of
generality that $\ell_0 \cap [y,z]=\{ y,z\}$. But then $[y,z]$
necessarily intersects one of the rays $r_1$ or $r_3$, say there
exists $z'\in [y,z] \cap r_3$.  Then we obtain two distinct geodesics
between $y$ and $z'$: one along $\ell'$ and the second along $[y,z]$.
\end{proof}

\begin{lemma} \label{Pasch} (Pasch axiom)
If $\Delta(x,y,z)$ is a geodesic triangle, $u \in [x,y]$, and $z \in
[x,v],$ then $[u,v] \cap [y,z] \ne \varnothing$ (see Figure~\ref{fig:axioms}(a)).
\end{lemma}

\begin{proof} The assertion is obvious if $u \in \{x,y\}$ or $v \in [y,z]$. So, let
$u \notin \{x,y\}$ and $v \notin [y,z]$. If there exists a point $v'
\in [u,v] \cap [x,z]$, then $z \in [u,v]$ by convexity of $[u,v] \cap
[x,v]$. So, we can further suppose that $[u,v] \cap [z,x] =
\varnothing$.

Note that $v \notin \Delta^*(x,y,z)$. Indeed, if $v \in
\Delta^*(x,y,z)$, consider any line $\ell$ extending
$[x,v]$. The ray from $\ell$ emanating from $v$ and not containing $x$ intersects
$[x,y]$ or $[y,z]$. But in this case $\ell \cap [x,y]$ or
$\ell \cap [y,z]$ is not convex.

Let $\ell$ be a line extending $[x,y]$.  Let $r_x$ and $r_y$ be the
two disjoint rays of $\ell$ emanating from $x$ and $y$, respectively.
$[u,v]$ intersects $r_y \cup [y,z] \cup [z,x] \cup r_x$. Because
$[u,v] \cap l$ is convex and does not contain $x$ nor $y$, $[u,v] \cap
(r_x \cup r_y) = \varnothing$. Therefore $[u,v] \cap [y,z] \neq \varnothing$.
\end{proof}

For two distinct points $x,y\in S,$ let $C(x,y)=\{ z\in S: x\in
[y,z]\}$ and $C(y,x)=\{ z\in S: y\in [x,z]\};$ we will call the
sets $C(x,y)$ and $C(y,x)$ {\it cones}. Since $S$ satisfies the
geodesic extension property, the set $C(x,y)\cup [x,y]\cup C(y,x)$ can
be equivalently defined as the union of all geodesic lines extending $[x,y]$.

\begin{lemma} \label{cones}
$C(x,y)$ and $C(y,x)$ are convex and closed subsets of $S$.
\end{lemma}

\begin{proof}
Let $u,v \in C(x,y)$ and $w \in [u,v], w \notin \{u,v\}$. By Pasch
axiom, there exists a point $w' \in [y,w] \cap [u,x]$. Since $w' \in
[u,x] \subset C(x,y)$, we conclude that $x \in [w',y] \subset [w,y]$,
whence $w \in C(x,y)$. This shows that $C(x,y)$ is convex. To show that
$C(x,y)$ is closed, let $\{u_i\}$ be a sequence of points of $C(x,y)$
converging to a point $u \in S$. Since $\{d(y,u_i)\}$ converges to
$d(y,u)$, $\{d(x,u_i)\}$ converges to $d(x,u)$, and $d(y,x) + d(x,u_i)
= d(y,u_i),$ we conclude that $d(y,x) + d(x,u) = d(y,u)$, hence $u\in
C(x,y)$.
\end{proof}

\begin{lemma} \label{triangle}  $\Delta^*(x,y,z)$ is the convex hull of $\{x,y,z\}$.
\end{lemma}

\begin{proof}
First we show that $\Delta^*:=\Delta^*(x,y,z)$ is
convex. Suppose by way of contradiction that for two points $p,q \in
\Delta^*$, $[p,q]$ contains a point $s\in S \setminus
\Delta^*$. Without loss of generality, we can assume that $[p,q] \cap
\Delta^* = \{p,q\}$ and that $p \in [x,y]$, $q \in [x,z]$. Let
$\ell$ be a line extending $[x,y]$. By convexity of closed halfplanes,
$[p,q]$ and $z$ must be in the same halfplane defined by $\ell$, thus, since
$[p,q] \cap \Delta^*=\{p,q\}$, the only possible positions for $p$ are $p \in \{x,
y\}$. Similarly, $q \in \{x, z\}$. This implies that $s \in \Delta \subset
\Delta^*$, a contradiction.

Now, we show that $\Delta^* \subseteq \conv(x,y,z)$. Notice that
$\Delta(x,y,z) \subseteq \conv(x,y,z)$. Let $w \in \Delta^* \setminus
\Delta(x,y,z)$. Let $\ell$ be any line containing $w$. The two rays of
$\ell$ emanating from $w$ must each intersect $\Delta$, say in $u$ and $v$
respectively. Hence $w \in [u,v] \subset \conv(x,y,z)$.
\end{proof}

\subsection{Geodesic convexity}\label{sec:convexity}
In this subsection, we establish some elementary properties of convex sets in Busemann
surfaces $(S,d)$.

\begin{lemma} \label{Peano} (Peano axiom)
 If $\Delta(x,y,z)$ is a geodesic triangle, $p\in [x,y],$ $q\in
[x,z],$ and $u\in [p,q],$ then there exists a point $v\in [y,z]$ such
that $u\in [x,v]$. (see Figure~\ref{fig:axioms}(b))
\end{lemma}

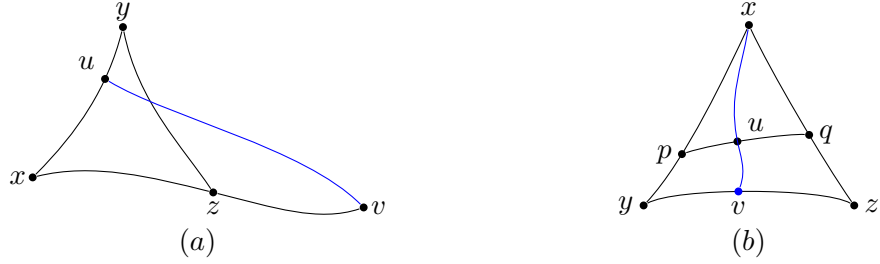
\begin{figure}
\begin{center}
\begin{tabular}{cc}
\begin{minipage}{0.45\textwidth}
\begin{center}
\begin{tikzpicture}[x=0.2cm,y=0.2cm,>=latex]
\vertex{black} (nx28y11) at (28,11) {};
\vertex{black} (nx18y12) at (18,12) {};
\vertex{black} (nx6y13) at (6,13) {};
\vertex{black} (nx12y23) at (12,23) {};
\draw[black] (nx18y12) .. controls (22,11) and (25,10) .. (nx28y11);
\draw[black] (nx18y12) .. controls (14,13) and (10,14) .. (nx6y13);
\draw[black] (nx12y23) .. controls (13,18) and (16,15) .. (nx18y12);
\draw[black] (nx6y13) .. controls (9,16) and (11,19) .. (nx12y23)
  node[circle,inner sep = 0pt,minimum size =3pt,fill = black,pos=0.7] (u) {};
\draw[blue] (nx28y11) .. controls (24,15) and (15,17) .. (u);
\draw (u) node[anchor = south east] {$u$};
\draw (29,11) node {$v$};
\draw (18,11) node {$z$};
\draw (5,13) node {$x$};
\draw (12,24) node {$y$};
\end{tikzpicture}
\end{center}
\end{minipage} &
\begin{minipage}{0.45\textwidth}
\begin{center}
\begin{tikzpicture}[x=1.4cm,y=1.2cm,>=latex]
\vertex{black} (x) at (1,2) {};
\vertex{black} (y) at (0,0) {};
\vertex{black} (z) at (2,0) {};

\draw[black] (x) .. controls (0.8,1.5) and (0.4,0.5) .. (y)
 node[circle,inner sep = 0pt,minimum size =3pt,fill = black,pos=0.7] (p) {};
\draw[black] (x) .. controls (1.2,1.5) and (1.7,0.5) .. (z)
 node[circle,inner sep = 0pt,minimum size =3pt,fill = black,pos=0.6] (q) {};
\draw[black] (y) .. controls (0.3,0.2) and (1.7,0.2) .. (z)
 node[circle,inner sep = 0pt,minimum size =3pt,fill = blue,pos=0.46] (v) {};
\draw[black] (p) .. controls (0.7,0.7) and (1.4,0.8) .. (q)
 node[circle,inner sep = 0pt,minimum size =3pt,fill = black,pos=0.4] (u) {};

\draw[blue] (x) to[out=260,in=100] (u) to[out=280,in=70] (v);

\draw (x) node[anchor=south] {$x$};
\draw (y) node[anchor=east] {$y$};
\draw (z) node[anchor=west] {$z$};
\draw (p) node[anchor=east] {$p$};
\draw (q) node[anchor=west] {$q$};
\draw (u) node[anchor=south west] {$u$};
\draw (v) node[anchor=north] {$v$};
\end{tikzpicture}
\end{center}
\end{minipage}\\
$(a)$ & $(b)$
\end{tabular}
\end{center}

\caption{Pasch and Peano axioms.} 
\label{fig:axioms}
\end{figure}

\begin{proof}
We can suppose that $u \notin \Delta(x,y,z),$ otherwise the result is
obvious. By Lemma \ref{triangle}, the point $u$ belongs to
$\Delta^*(x,y,z).$ Let $\ell$ be a geodesic line extending $[x,u]$ and
let $r_{u}$ be the ray of $\ell$ emanating from $u$ and not passing via
$x$. Necessarily $r_u$ will intersect one of the sides of
$\Delta^*(x,y,z)$.  Since $u\notin [x,y]\cup [x,z]$, $x\in \ell\setminus r_u$, and $S$
is uniquely geodesic, $r_u$ necessarily intersects $[y,z]$ in a point
$v$. Then $u \in [x,v],$ and we are done.
\end{proof}

It is well known \cite{VdV} that Peano axiom is equivalent to the following property,
called {\it join-hull commutativity}: for any convex set $A$ and any
point $p \notin A,$ conv$(p \cup A) = \cup \{ [p,x]: x \in A\}$. As a
consequence, we obtain:

\begin{lemma} \label{jhc} $S$ is join-hull commutative.
\end{lemma}

We continue by establishing that the Carathéodory number of $S$ equals 3:

\begin{lemma} \label{Caratheodory}
For any finite set $Q$ of $S$, $\conv(Q)=\cup\{ \Delta^*(x,y,z): x,y,z\in Q\}.$
\end{lemma}

\begin{proof}
We proceed by induction on $n=|Q|.$ Let $x \in Q$ and suppose
that the assertion holds for the set $Q^x = Q \setminus
\{x\}$. Let $K^x = \conv(Q^x)$. Pick $p\in
\conv(Q).$ If $p \in K^x,$ we are done by induction
assumption. So, let $p \notin K^x$. Since
$\conv(Q) = \conv(x \cup K^x)$, by join-hull commutativity
there exists a point $q\in K^x$ such that $p \in [x,q]$. By induction
assumption, there exists three points $y,z,v \in Q^x$ such that $q \in
\Delta^* = \Delta^*(y,z,v)$. The geodesic segment $[x,q]$ necessarily
intersects one of the sides $[y,z],[z,v],[v,y]$ of $\Delta^*$, say there exists
$q' \in [x,q]\cap [y,z]$. But then $p \in [x,q']$ and we conclude that
$p \in \conv (x,y,z)$ with $x,y,z \in Q$.
\end{proof}

Let $Q$ be a nonempty finite set of points of $S$ and $K:=\conv(Q)$. Let
$Q_0$ denote the set of all points $u\in Q$ such that $u$ does not
belong to any triangle $\Delta^*(x,y,z)$ with $x,y,z \in Q$ and $u\ne
x,y,z$. By Lemma \ref{Caratheodory}, $Q_0 \ne \varnothing$, moreover
$\conv(Q_0) = \conv(Q).$ We call the points of $Q_0$ {\it
extremal points} of $Q$ (or of $K$). A line $\ell$ is called a {\it
tangent line} (or, simply a {\it tangent}) of $K$ if $K \cap
\ell \ne \varnothing$ and $K$ is contained in one of the closed halfplanes
defined by $\ell$.  A geodesic segment $[x,y]$ is called an {\it edge}
of $K$ if $x,y\in Q$ and some line $\ell$ extending $[x,y]$ is a
tangent of $K$. Clearly, each edge of $K$ belongs to the boundary of
$K$. A geodesic line $\ell$ is called a {\it bitangent} of two
disjoint convex sets $K'$ and $K''$ if $\ell$ is a tangent line of
$K'$ and $K''$. A bitangent $\ell$ of $K',K''$ is called an {\it inner
bitangent} if $K'$ and $K''$ belong to different closed halfplanes
defined by $\ell$.

\begin{lemma} \label{edges_general_position}
If $Q$ is a nonempty finite set in general position of $S,$ then $[p,q]$ is an
edge of $K := \conv(Q)$ if and only if any line $\ell$ extending
$[p,q]$ is a tangent of $K$. If $(Q',Q'')$ is a bipartition of $Q$
such that the convex hulls $K' = \conv(Q'), K'' = \conv(Q'')$
are disjoint and $\ell$ is an inner bitangent of $K',K''$ such that
$p' \in Q' \cap \ell$ and $p'' \in Q'' \cap \ell,$ then any geodesic line
extending $[p',p'']$ is an inner bitangent of $K',K''$.
\end{lemma}

\begin{proof}
Let $[p,q]$ be an edge of $K$ and let $\ell$ be a tangent of $K$
extending $[p,q]$. Suppose by way of contradiction that some line
$\ell'$ extending $[p,q]$ is not a tangent of $K$, i.e., two points
$x$ and $y$ of $Q$ belong to complementary halfplanes defined by
$\ell'$, say $y$ is in the region delimited by the rays of $\ell$ and
$\ell'$ emanating from $p$ that do not contain $q$. Then consider
$[y,q]$, by convexity of the halfplanes delimited by $\ell$ and
$\ell'$, $p \in [y,q]$, hence $y \in C(q,p)$, contrary to the
assumption that the points of $Q$ are in general position.

Analogously, if $\ell$ is an inner bitangent of $K',K'$ as defined in
the lemma, but a line $\ell'$ extending $[p',p'']$ is not an inner
bitangent, we will conclude that some point $x \in Q$, $x \notin \{p',p''\}$
belongs to one of the cones $C(p',p'')$ or $C(p'',p'),$ contrary to
the assumption that the points of $Q$ are in general position.
\end{proof}

\begin{lemma} \label{decomposition_convex_hull}
For any finite set $Q$ of $S,$ the set $Q_0$ of extremal points of
$K:=\conv (Q)$ admits a circular order $\pi=(p_{i_1},\ldots
p_{i_k}),$ such that the geodesic segments $[p_{i_j},p_{i_{j+1
(\textrm{mod} ~k)}}]$ are edges of $K$ and the boundary of $K$ is the
union of these edges $[p_{i_j},p_{i_{j+1 (\textrm{mod} ~k)}}]$.
\end{lemma}

\begin{proof}
We proceed by induction on $|Q|$. For a point $x \in Q,$ let
$Q^x = Q \setminus \{x\}$ and $K^x = \conv(K^x)$.  If for some $x \in
Q,$ $x \in K^x,$ then $x \notin Q_0$ and
$\conv(Q)=\conv(Q^x),$ so we can apply the induction
hypothesis to the set $Q^x$. Thus we can assume that all points of $Q$
are extremal.

To define the required circular order $\pi$ on $Q,$ for each point
$x \in Q$ we have to define its two neighbors in $\pi$.  Consider the
set $Q^x$. Then obviously $Q^x_0 = Q^x$ and, by the induction
assumption, the points of $Q^x$ admit the required circular order
$\pi'$. We will prove now that we can find two consecutive points
$y,z \in Q^x$ of $\pi'$ such that $[x,y]$ and $[x,z]$ are edges of
$K$. Then the circular order $\pi$ is obtained from $\pi'$ by
inserting the point $x$ between $y$ and $z$.

For two distinct points $u,v \in Q^x$, let $R(u,v)$ be the closed
region of $S$ comprised between two rays of $\ell'$ and $\ell''$ from
$x$ extending $[x,u]$ and $[x,v],$ respectively. Let $n(u,v)$ be the
number of points of $Q \cap R(u,v)$.
Notice that for any point $w \in Q^x,$
we have $w \in R(u,v)$ if and only if $[x,w] \cap [u,v] \ne
\varnothing$.

Let $y,z$ be a pair of points of $Q^x$ for which $n(y,z)$ is
maximal. We claim that $n(y,z) = |Q|$. Suppose by way of contradiction
that there exists a point $y'\in Q^x$ such that some line $\ell'$
extending $[x,y]$ separates $y'$ from $z$. We claim that
$n(y',z) > n(y,z)$. Let $\ell$ be a geodesic line extending $[x,y']$ and
defining $n(y',z)$. First, since $\ell'$ separates $y'$ and $z,$ we
have $[y',z] \cap \ell' \ne \varnothing$.  Since $y\notin
\Delta^*(x,y',z)$ because the points of $Q$ are extremal, this implies
that $[y',z] \cap [x,y] \ne \varnothing$,
i.e., $y \in R(y',z).$ Let $u \in [y',z] \cap [x,y]$. Pick any point
$w \in R(y,z)$. Then there exists $w' \in [y,z] \cap [w,x].$ Since
$[u,z] \cap [w',x] \ne \varnothing$ by Pasch axiom applied to the
triangle $\Delta(x,y,w')$, we conclude that
$[w,x]$ intersects $[y',z],$ whence $w \in R(y',z).$ Since $y'\in
R(y',z) \setminus R(y,z),$ we deduce that $n(y',z) > n(y,z),$ contrary to
the maximality choice of the pair $y,z$. This proves that
$n(y,z) = |Q|$. Therefore for the lines $\ell', \ell''$ extending
$[x,y]$ and $[x,z],$ the set $Q$ is contained in the convex set
$R(y,z)$ bounded by $\ell'$ and $\ell''$. Hence $\ell',\ell''$ are
tangents of $K = \conv(Q),$ showing that $[x,y],[x,z]$ are edges
of $K$.

Note that $[y,z]$ is an edge of $K^x = \conv(Q^x)$ (i.e., $y,z$ are
consecutive in $\pi'$), otherwise $\Delta^*(x,y,z)$ will contain yet
another point of $Q$, contrary to the assumption that all points of
$Q$ are extremal. It remains to show that any edge $[u,v]$ of $K^x$
different from $[y,z]$ is also an extremal edge of $K$. From the choice
of $[y,z],$ we conclude that $[x,u]\cap [y,z]\ne \varnothing$ and
$[x,v]\cap [y,z]\ne \varnothing$. Let $u' \in [x,u]\cap [y,z]$. Pick any
line $\ell$ extending $[u,v]$; $\ell$ is a tangent of $K^x$. If $\ell$
is not a tangent of $K$, then $\ell$ intersects $[x,y]$ and
$[x,z]$. Consequently, $[x,u']$ intersects $\ell$, and thus $\ell \cap
[x,u]$ is not convex, a contradiction. Hence $[u,v]$ is also an edge
of $K$.
\end{proof}

\begin{lemma}\label{convex_and_point}
Let $A \subset S$ and $q \notin \conv(A)$. Let $\ell$ and $\ell'$ be
two lines containing $q$ and intersecting $\conv(A)$. Orient $\ell$ and
$\ell'$ from $q$ to $\conv(A)$, and let $H^+_\ell$, $H^+_{\ell'}$ be the two
halfplanes to the right of $\ell$ and $\ell'$, respectively. Then,
either $A \cap H^+_\ell \subset H^+_{\ell'}$ or $A \cap H^+_{\ell'} \subset H^+_{\ell}$.
\end{lemma}

\begin{proof}
Suppose by way of contradiction that there exist $p \in \left(H^+_l
\cap A\right) \setminus H^+_{\ell'}$ and $p' \in \left(H^+_{\ell'}
\cap A\right) \setminus H^+_{\ell}$. Consider the geodesic triangle
$\Delta(p,p',q)$. By assumption $q \notin [p,p']$. Let $u \in [p,p']
\cap \ell$ and $u' \in [p,p'] \cap \ell'$.  We have $u \neq p'$, $u'
\neq p$, and $u, u' \in \conv(A)$. By Lemma~\ref{triangle}, $[p,u]
\subset \Delta^*(q,p,p')$ and $[p,u'] \subset \Delta^*(q,p,p')$. This
contradicts the facts that $p'$ is on the left of $\ell$ and on the
right of $\ell'$ and that $p$ is on the left of $\ell'$ and on the
right of $\ell$.
\end{proof}

\begin{lemma} \label{bitangent}
If $(Q',Q'')$ is a proper bipartition of a finite set $Q$ of $S$ such that the
convex hulls $K'=\conv(Q'), K''=\conv(Q'')$ are disjoint, then there
exists four (not necessarily distinct) points $p', q' \in Q'$ and
$p'', q'' \in Q''$ and two inner bitangents extending $[p',p'']$ and
$[q',q'']$ respectively. Any inner bitangent of $K',K''$  extends $[p',p'']$ or
$[q',q'']$.
\end{lemma}

\begin{proof}
For a line $\ell$ extending a segment $[p',p'']$ with $p' \in Q'$,
$p'' \in Q''$, orient $\ell$ from $p'$ to $p''$, and denote by $H^+_l$
the closed halfplane delimited by $\ell$ on the right of $\ell$, and
by $H^-_\ell$ the closed halfplane on the left of $\ell$.

Let $p' \in Q'$ and $p'' \in Q''$ be two points such that for some
line $\ell'$ extending $[p',p'']$ the value of $|Q' \cap H^-_{\ell'}|
+ |Q''\cap H^+_{\ell'}| = n_1(p',p'')$ is as large as possible. We
assert that $\ell'$ is an inner bitangent of $K'$ and $K''$. Suppose
not: then either there exists $q \in Q'' \cap \mathring{H}^-_{\ell'}$ or $p \in
Q' \cap \mathring{H}^+_{\ell'}$, say the first. We will find a pair $s',s''$ with
$n_1(s',s'') > n_1(p',p'')$, leading to a contradiction with the
maximality choice of the pair $p',p''$.

Pick any line $\ell$ extending $[p',q]$. By
Lemma~\ref{convex_and_point} using point $p'$, $Q'' \cap H^+_{\ell'}
\subset Q'' \cap H^+_{\ell}$.  If $Q' \cap H^-_{\ell'} \subset Q' \cap
H^-_{\ell}$, then since $q \in Q'' \cap H^+_{\ell} \setminus Q'' \cap
H^+_{\ell'}$, we conclude that $n_1(p',q) > n_1(p',p'')$.

So, suppose that there exists a point $p \in Q' \cap H^-_{\ell'}$ such
that $p \in \mathring{H}^+_{\ell}$. Let $p \in Q'\cap H^-_{\ell'} \cap
\mathring{H}^+_{\ell}$ be such that for some line $\ell''$ extending
$[p,q]$ the value of $|Q' \cap H^-_{\ell''}|$ is as large as possible.
Suppose there exists $p^* \in Q'\cap H^-_{\ell'} \cap
\mathring{H}^+_{\ell''}$. Let $\ell^*$ be any line extending
$[q,p^*]$.  Then, by Lemma~\ref{convex_and_point} applied to $Q'$ and
$q$ for the lines $\ell''$ and $\ell^*$, we get that $|Q' \cap
H^-_{\ell''}| < |Q' \cap H^-_{\ell }|$, contradicting our choice of
$p$.

By Lemma~\ref{two-lines}, we can assume that $H^-_{\ell'} \cap
\mathring{H}^+_{\ell}$ is contained in the region delimited by the
rays of $\ell$ and $\ell'$ emanating from $p'$ containing respectively
$q$ and $p''$. Consequently, $p \in Q' \cap H^-_{\ell'} \cap
\mathring{H}^+_{\ell}$, and thus, $p$ is in the triangle
$\Delta^*(p',q,p'')$. Consider the ray of $\ell''$ emanating from $p$
not containing $q$. In $\Delta^*(p',q,p'')$, it intersects
$[p',p'']$. Let $u \in [p',p''] \cap \ell''$; note that $u \notin
K''$, since it would imply that $p\in K''$.  By
Lemma~\ref{convex_and_point} applied to $u$, we have that $Q'' \cap
H^+_{\ell'} \subsetneq H^+_{\ell''}$. Consequently, from our choice of
$p$, we have $n_1(p,q) > n_1(p',q')$, contradicting the choice of $p'$
and $q'$.

Analogously, taking two points $q' \in Q'$ and $q'' \in Q''$ such that
for some geodesic line $\ell''$ extending $[q',q'']$ the value of $|Q'
\cap H^+_{\ell''}| + |Q'' \cap H^-_{\ell''}| = n_2(q',q'')$ is
maximum, we will obtain that $n_2(q',q'') = |Q'| + |Q''|$ and thus
$\ell''$ is an inner bitangent.

Finally, if $\ell$ is a third inner
bitangent extending an other segment $[p,q]$ with $p \in Q'$ and $q \in Q''$,
then $\ell$ will necessarily intersect twice one of the lines $\ell'$
or $\ell'$. This establishes that any inner bitangent of $K',K''$
either extends $[p',p'']$ or $[q',q'']$.
\end{proof}

Two sets $K',K''$ of $S$ are called {\it line-separable} if there
exists a geodesic line $\ell$ such that $K'\subset
\mathring{H}^-_{\ell}$ and $K''\subset \mathring{H}^+_{\ell}.$ We will
call two sets $K',K''$ {\it weakly line-separable} if there exists a
line $\ell$ such that $K'\subseteq H^-_{\ell}$ and $K''\subseteq
H^+_{\ell}.$ In the first case, we will say that the line $\ell$ {\it
separates} $K'$ and $K''$ and in the second case that $\ell$ {\it
weakly separates} $K'$ and $K''$.

\begin{lemma} \label{line-separable}
If $(Q',Q'')$ is a bipartition of a finite set $Q$ of $S$ such that the
convex hulls $K' = \conv(Q'), K'' = \conv(Q'')$ are disjoint, then
$K'$ and $K''$ are line-separable.
\end{lemma}

\begin{proof}
Let $\ell'$ and $\ell''$ be two inner bitangents of $K',K''$ defined
as in Lemma~\ref{bitangent}: $\ell'$ extends $[p',p'']$ and $\ell''$
extends $[q',q'']$ with $\{p',q'\} \subset Q'$ and $\{p'',q''\} \subset Q''$ (it may
happen that $p' = q'$ or $p'' = q''$).  Let $[p',p''] \cap [q',q''] = [u,v]$,
where $u \in [p',v]$.

\medskip\noindent
{\bf Case 1.} $u\in K'$ and $v\in K''$.

\medskip\noindent Since $K'$ and $K''$ are disjoint, $u \ne v$.  Let
$w$ be an arbitrary point of $[u,v]$ not belonging to $K'$ and $K''$
and let $\epsilon > 0$ be such that $B(w,\epsilon) \cap (K' \cup K'')
= \varnothing$ (it exists since $K', K''$ are closed by
Lemma~\ref{Caratheodory}). Let $x$ and $y$ be two points from
different connected components of $B(w,\epsilon)\setminus [u,v].$ Let
$\ell$ be a geodesic line extending $[x,y]$. By convexity of
$B(w,\epsilon)$, $[x,y] \subset B(w,\epsilon)$ and $\varnothing \neq
\ell \cap [u,v] \subset B(w,\epsilon)$. Hence the two disjoint rays of
$\ell$ with endpoints $x$ and $y$ are disjoint from $\ell'$ and
$\ell''$, thus $\ell$ separates $K'$ and $K''$.

\medskip\noindent
{\bf Case 2.} $u \notin K'$ or $v \notin K'',$ say $v \notin K''$.

\medskip\noindent Let $\epsilon > 0$ such that $B(v,\epsilon) \cap (K'
\cup K'') = \varnothing$. Let $w$ be a point in $B(v,\epsilon/2)
\setminus (\ell' \cup \ell'')$, in the same connected component as
$K''$ in $S \setminus (\ell' \cup \ell'')$. Because $w \notin K''$, by
Lemma~\ref{decomposition_convex_hull} there are two geodesic
$[q_1,w]$, $[w,q_2]$ with $\{q_1,q_2\} \subset Q''$, extendable in two
distinct lines $\ell_1$ and $\ell_2$ respectively, both tangent to
$\conv(Q'' \cup \{w\})$. Denote by $H_1$ and $H_2$ the closed
halfplanes containing $K''$ delimited by $\ell_1$ and $\ell_2$
respectively.  For $i \in \{1,2\}$, the rays of $\ell_i$ must each
intersect a distinct side of the triangle $\Delta(v,p'',q'')$. Note
that if $\ell_i$ intersects $[p'',q'']$, then $[p'',q'']$ is an edge
of $K''$ and $\ell_i$ contains either $p''$ or $q''$. Consequently,
$\ell_i$ must intersects both $[v,p''] \subset \ell'$ and $[v,q'']
\subset \ell''$, and $v \notin {\ell}_i$. Clearly $K' \subset
S\setminus H_i$.

Choose $x \in (\mathring{H_1} \setminus H_2) \cap B(w,\epsilon/2)$ and
$y \in (\mathring{H_2} \setminus H_1) \cap B(w,\epsilon/2)$, clearly
$[x,y] \subset B(v,\epsilon)$. Moreover $\ell_1 \cap [x,y] \neq
\varnothing$ and $\ell_2 \cap [x,y] \neq \varnothing$. Then any line extending
$[x,y]$ separates $K'$ and $K''$.
\end{proof}

\begin{lemma} \label{separation}
If $(Q',Q'')$ is a bipartition of a finite set $Q$ in general position
of $S$ such that the convex hulls $K'=\conv(Q'), K''=\conv(Q'')$ are
weakly line-separable and $(K'\cap K'')\cap Q=\varnothing,$ then the
sets $Q'$ and $Q''$ are line-separable.
\end{lemma}

\begin{proof}
If $K'$ and $K''$ are disjoint, then the result follows from Lemma
\ref{line-separable}. So, let $K' \cap K'' \ne \varnothing$ and let
$\ell_0$ be a line weakly-separating $K'$ from $K''$. We denote the
two halfplanes delimited by $\ell_0$ as $H^+$ and $H^-$ such that $K'
\subset H^-$ and $K'' \subset H^+$. Then $K' \cap K'' \subset \ell_0$,
thus $K' \cap K'' = [p,q]$ for some points $p,q \in K' \cap
K''$. Since $(K' \cap K'') \cap Q = \varnothing$, by Lemma
\ref{decomposition_convex_hull} there exists an edge $[p',q']$ of $K'$
and an edge $[p'',q'']$ of $K''$ such that $[p,q] = [p',q'] \cap
[p'',q'']$ and $\{ p,q\}\cap \{ p',q',p'',q''\}=\varnothing$ (because
the points of $Q$ are in general position). Let
$\ell'$ and $\ell''$ be two tangents of $K'$ and $K''$ extending
$[p',q']$ and $[p'',q'']$, respectively ($\ell_0$ may coincide with one
of these tangents). Notice that $\ell'$ and $\ell''$ also separate
$K'$ and $K''$, as by convexity $\ell' \cap K'' = [p,q] = \ell'' \cap
K'$. So, denote by $H^-_{\ell'}$ and $H^+_{\ell'}$ the halfplanes
delimited by $\ell'$, and $H^-_{\ell''}$ and $H^+_{\ell''}$ the
halfplanes delimited by $\ell''$, such that $K' \subset H_{\ell'}^-
\cap H_{\ell''}^-$ and $K'' \subset H_{\ell'}^+ \cap H_{\ell''}^+$.

Let $\epsilon>0$ be such that $\min \{d(p,p'), d(p,p''), d(q,q'),
d(q,q'')\} > \epsilon.$ Pick two points $x$ and $y$ in
$\mathring{H}_{\ell'}^+ \cap \mathring{H}_{\ell''}^-$ so that $x \in
B(p,\epsilon)$ and $y \in B(q,\epsilon)$. Let $\ell$ be any geodesic
line extending $[x,y]$.  We assert that $\ell$ line-separate $Q'$ and
$Q''$. Since $\{x,y\} \subset \mathring{H}_{\ell'}^+ \cap
\mathring{H}_{\ell''}^-$, from Lemma \ref{closed_halfplane} we
conclude that $[x,y] \subset H_{\ell'}^+ \cap H_{\ell''}^-$. Since $K'
\subset H_{\ell'}^-$ and $K''\subset H_{\ell''}^+,$ we conclude that
$\ell \cap K'\subset \ell' \cap K' = [p',q']$ and $\ell \cap K''
\subset \ell'' \cap K''=[p'',q'']$. In particular, $[p,q] \subset
\ell$ and $K' \subset H^-_{\ell}, K''\subset H^+_{\ell}.$ Therefore,
if $\ell$ does not line-separate $Q'$ and $Q'',$ there exists a point
of $Q'$ or $Q''$ on $\ell$, say $Q'$, hence on $\ell \cap K' \subset
[p',q']$. This contradicts the general position assumption for $Q$.
\end{proof}

\section{Crofton formula (after R. Alexander)}\label{sec:crofton}

In this section we will prove a Crofton formula for finite subsets in
general position of Busemann surfaces. Our proof
uses the convexity results of previous section and follows closely  and
generalizes an analogous result of Alexander \cite{Al}.  

Let $(S,d)$ be a Busemann surface and let $Q=\{ p_1,\ldots, p_n\}$ be
a finite set of distinct points of $S$ in general position.  A line
$\ell$ of $S$ is said to {\it separate} the points of $Q$ if (i) none
of the points of $Q$ lie on $\ell$, and (ii) each of the open
halfplanes determined by $\ell$ contains at least one of the points of
$Q$. Two separating lines are called {\it equivalent} if each
separates $Q$ into the same pair of sets.  Let $L_1,\ldots, L_m$ be
the equivalence classes of lines of $S$ separating the points of $Q$.
For a line $\ell$ separating the points of $Q$ into two nonempty
subsets $Q',Q'',$ let $K'=\conv(Q')$ and $K''=\conv(Q'').$
We will say that a pair of points $p_i,p_j\in Q$ constitutes an {\it
extremal segment} $[p_i,p_j]$ for the pair $(K',K'')$ (or $(Q',Q'')$)
if any line $\ell$ extending $[p_i,p_j]$
weakly separates $K'$ and $K''$.  We
will call an extremal segment $[p_i,p_j]$ {\it positive} if either
$p_i\in Q', p_j\in Q''$ or $p_i\in Q'', p_j\in Q'$ and we call
$[p_i,p_j]$ {\it negative} if either $p_i,p_j\in Q'$ or $p_i,p_j\in
Q''.$

By Lemma \ref{bitangent} (and similarly to the Euclidean plane) the
positive extremal segments correspond to the two pairs $[p',p'']$ and
$[q',q'']$ of segments with $p',q'\in K', p'',q''\in K''$ such that
any line extending one of these segments is an inner bitangent of $K'$
and $K''$.  On the other hand, since any line extending a negative
segment of $K'$ (or of $K''$) is a tangent of $K'$ (or of $K''$), the
negative extremal segments are edges of $K'$ or $K''$. In the
Euclidean plane, they are precisely the edges of $K'$ and $K''$ which
belongs to the triangles $\Delta^*(p',q',x)$ and $\Delta^*(p',q',x)$,
where $x=[p',p'']\cap [q',q''].$ We will show now that a similar
property holds for Busemann surfaces:

\begin{lemma} \label{negative_segment} An edge $[p,q]$ of $K'$ is a negative extremal segment of
$(K',K'')$ if and only if $p,q \in \Delta^*(p',q',x),$ where $x$ is any point of $[p',p''] \cap [q',q'']$.
\end{lemma}

\begin{proof}
Let $\ell',\ell''$ be two inner bitangents of $K',K''$ extending
$[p',p'']$ and $[q',q'']$, respectively. Let $A',A'',B',B''$ be the
four pairwise intersections of closed halfplanes defined by the lines
$\ell'$ and $\ell'',$ so that $K'\subset A'$ and $K''\subset A''$.
Pick any edge $[p,q]$ of $K'$ and let $\ell$ be any line extending
$[p,q].$ Then $\ell$ is a tangent of $K'$. From Lemma
\ref{decomposition_convex_hull}, either both $p,q$ belong to
$\Delta^*(p',q',x)$ or neither of $p$ or $q$ belong to the interior of
$\Delta^*(p',q',x)$.

First suppose that $p,q\in \Delta^*(p',q',x)$. Since $\ell$ is a
tangent of $K',$ $\ell$ does not intersect $[p',q']$. Hence, $\ell$
necessarily intersects the sides $[x,p']$ and $[x,q']$ of
$\Delta^*(p',q',x)$ and enters in the regions $B',B''$.  If $[p,q]$ is
not an extremal segment, then $\ell$ intersects $K''$ (and therefore
$A''$). This means that $\ell$ will intersect twice $\ell'$ or
$\ell'',$ a contradiction. This proves that any edge $[p,q]$ of $K'$
with $p,q\in \Delta^*(p',q',x)$ is a negative extremal segment.

Conversely, suppose that $[p,q]$ is an edge of $K'$ such that
$p,q\notin \Delta^*(p',q',x)\setminus \Delta(p',q',x)$. Suppose by way
of contradiction that $[p,q]$ is a negative extremal segment, i.e.,
$\ell$ weakly separates $K'$ from $K''$. Then necessarily $\ell$
intersects the boundary of $A'$ consisting of rays of $\ell',\ell''$,
say $\ell$ intersects $\ell'\cap A'\cap B'$.  If $\ell$ does not intersect
the interior of $B',$ we will conclude that $\ell$ passes via the
point $p'.$ By Lemma \ref{two-lines} there exists a line extending $[p,q]$
and passing via $p'$ and $q''$, contrary to the
assumption that the points of $Q$ are in
general position. If $\ell$ intersects the interior of $B'$ (i.e.,
$B'\setminus (\ell'\cup \ell'')$), then since $\ell$ weakly separates
$K'$ and $K''\subset A'',$ $\ell$ necessarily intersects a second
time the line $\ell'$ in $[p',q'']$, which is impossible.
Hence $[p,q]$ cannot be a
negative extremal segment.
\end{proof}


For an equivalence class $L_{t}$  of lines, $t=1,\ldots, m$, we will
denote by $Q'_{t}$ and $Q''_{t}$ the two subsets of $Q$ into which any
line $\ell$ of $L_{t}$ separates the points of $Q$ and by $K'_t$ and
$K''_t$ we denote their convex hulls. Let also $ES^+_t$ and $ES^-_t$
be the sets of positive and negative extremal segments of the pair
$(K'_t,K''_t)$. Finally,
set $$\sigma_t:=\sum_{[p_i,p_j]\in ES^+_t}
d(p_i,p_j)-\sum_{[p_i,p_j]\in ES^-_t} d(p_i,p_j).$$
We will show  that for any equivalence class $L_{t}$, $\sigma_t\ge 0$.
Let $p',q'\in K'_t$ and $p'',q''\in K''_t$ such that $[p',q']$ and
$[p'',q'']$ are the positive extremal segments of $K_t',K_t''$. Let $x\in
[p',p''] \cap [q',q'']$. Let $$\sigma'_t = d(x,p') + d(x,q') - \sum \{
d(p,q): [p,q] \in ES^-_t, p,q \in K_t'\} $$ and
$$\sigma''_t = d(x,p'') + d(x,q'') - \sum \{ d(p,q): [p,q] \in ES^-_t, p,q \in K_t''\}.$$
Since $\sigma_t = \sigma'_t + \sigma''_t$, it suffices to show
that $\sigma'_t \ge 0$ and $\sigma''_t \ge 0$. These inequalities are immediate consequences
of the following result:

\begin{lemma}
Let $Q'$ be a finite set of points of $S$ in general position and let
$K' = conv(Q')$.  For any $x \notin K'$, for all $p',q' \in Q'$ such
that $[x,p']$ and $[x,q']$ are edges of $\conv(Q' \cup \{x\})$,
$$
d(x,p')+d(x,q') \geq \sum \{d(p,q):[p,q] \mbox{ is an edge of } K
\mbox{ and } p,q \in \Delta^*(x,p',q')\}.
$$
\end{lemma}

\begin{proof}
If $p' = q'$, the lemma trivially holds. In the following,
we assume that $p' \neq q'$.  We prove the lemma by induction on the
number of segments in the set $ES^-(Q',x,p',q') = \{[p,q] : [p,q] \mbox{ is an
  edge of } K' \mbox{ such that } p,q \in \Delta^*(x,p',q')\}$. Note
that since the points  of $Q'$ are in general position, there exists at
most one $p'' \in Q' \cap [x,p']$. If $p''$ exists, then $[p',p'']$ is an
edge of $K'$, $d(x,p') = d(x,p'')+d(p'',p')$, and
$ES^-(Q',x,p',q')=ES^-(Q',x,p'',q')\cup\{[p',p'']\}$. By induction, we
immediately get the result. Thus, we can further assume that $Q' \cap [x,p'] =Q'
\cap [x,q'] = \varnothing$.

By Lemma \ref{decomposition_convex_hull}, the segments of
$ES^-(Q',x,p',q')$ form a path constituted by all edges of $K'$
located in $\Delta^*(p',q',x)$.  Let $[p_0,q_0] \in ES^-(Q',x,p',q')$
with $p_0=p'$. If $q_0=q'$, then $[p_0,q_0]$ is the unique segment in
$ES^-(Q',x,p',q')$ and by triangle inequality $d(p_0,q_0) = d(p',q')
\le d(p',x)+d(x,q')$. So, let $q_0 \ne q'$. Since the points of $Q'$
are in general position and $Q' \cap [x,p'] = Q' \cap [x,q'] =
\varnothing$, $q_0$ belongs to the interior of
$\Delta^*(p',q',x)$. Let $\ell$ be a line extending $[p_0,q_0]$. Then
$\ell$ necessarily intersects the segment $[x,q']$ in a point
$x'$. Note that all segments of $ES^-(Q',x,p',q')$ except $[p_0,q_0]$
are located in the triangle $\Delta^*(q_0,q',x')$.

By induction assumption, $d(q_0,x') + d(q',x') \ge \sum \{d(p,q) :
[p,q] \in ES^-(Q',x',q_0,q') \}$. On the other hand, $d(p_0,q_0) +
d(q_0,x') = d(p',x') \leq d(p',x) + d(x,x')$ by triangle inequality.
Putting these two inequalities together and taking into account that
$p' = p_0,$ we obtain:
\begin{align*}
d(p',x) + d(x,q') &= d(p',x) + d(x,x') + d(x',q') \ge d(p',x') + d(x',q') \\
                  &= d(p_0,q_0) + d(q_0,x') + d(x',q') \\
&\ge d(p_0,q_0) + \sum \{d(p,q) : [p,q] \in ES^-(Q',x',q_0,q') \} \\
                  &= \sum \{ d(p,q): [p,q] \in ES^-(Q',x',p',q')\}.
\end{align*}

\end{proof}

\begin{lemma} \label{cases2-4} Let $Q$
be a finite set of points in general position of a Busemann surface
$(S,d).$ Let $q,r\in Q$ two distinct points of $Q$ and $\ell$ an arbitrary line
passing via $[q,r]$. Let $\ell$ separates
the set $Q-\{ q,r\}$ into (possibly empty) sets $A$ and $B$. For any pair $(X,Y)$ among
 $(A \cup \{r\},B \cup \{q\})$, $(A \cup \{q\},B \cup \{r\})$, $(A,B \cup \{q,r\})$, $(A\cup\{ q,r\},B)$,
if $X\ne\varnothing$ and $Y\ne\varnothing$, then $X$ and $Y$ are line-separable.
\end{lemma}

\begin{proof}  Let $\mathring{H}^-_{\ell}, \mathring{H}^+_{\ell}$ denote the open halfplanes containing
respectively the sets $A$ and $B$. Since the closed halfplanes
$H^-_{\ell}$ and $H^+_{\ell}$ are convex and $A \cup \{q,r\} \subset
H^-_{\ell}, B \cup \{q,r\} \subset H^+_{\ell}$, the sets $\conv(A),
\conv(A \cup \{r\}), \conv(A \cup \{q\})$, and $\conv(A \cup \{q,r\})$
are contained in $H^-_{\ell}$ while the sets $\conv(B), \conv(B \cup
\{q\}), \conv(B \cup \{r\})$, and $\conv(B \cup \{q,r\})$ are
contained in $H^+_{\ell}$. Hence the convex hulls of the four pairs of sets from
the statement are weakly-separated by the line $\ell$;
thus their intersections are contained in the line $\ell$, consequently,
$\conv(X)\cap \conv(Y)\cap Q\subseteq \{ q,r\}$.

If $r\in \conv(X)\cap \conv(Y),$ then $r\in \conv(A\cup \{ q\})$ or $r\in \conv(B\cup \{ q\})$,
say the first. Since $r$ will belong
to the boundary of $\conv(A\cup \{ q\}),$ by Lemma
\ref{decomposition_convex_hull} there exist  points $q',q''\in
A\cup\{q\}\subset Q$ such that $r\in [q',q''],$ a contradiction with the fact
that the points of $Q$ are in general position. Hence $\conv(X)\cap \conv(Y)\cap Q=\varnothing.$ By Lemma \ref{separation},
the convex sets $X$ and $Y$ are line-separable.
\end{proof}

Here is the main
result of this section:

\begin{proposition} \label{crofton} (Crofton formula)
Let $Q=\{ p_1,\ldots, p_n\}$
be a finite set of points in general position of a Busemann surface
$(S,d).$ Then for any two points $p_i,p_j\in Q,$ we have
\begin{align*}
2d(p_i,p_j) = \sum\{ \sigma_t: t\in [1,m] \mbox{ and any line of } L_t \mbox{ intersects the segment } [p_i,p_j]\}. \tag{1}
\end{align*}
In particular, the finite metric space $(Q,d)$ is $l_1$-embeddable.
\end{proposition}

\begin{proof}
The proof uses double counting and closely follows the proof of Lemma
1 of \cite{Al}. We will show that each segment $[q,r]$ with $q,r \in Q$
occurs as an extremal segment in such a way that $+d(q,r)$ appears
exactly the same number of times as $-d(q,r)$ on the right side of the
equation (1), unless $[q,r] = [p_i,p_j]$, in which case $+ d(q,r)$
appears twice. We distinguish four cases:

\medskip\noindent {\bf Case 1:} The points $p_i, p_j, q, r$ are distinct
and no line $\ell$ containing $[q,r]$ cut the segment $[p_i,p_j]$.
\medskip

We assert that in this case $\pm d(q,r)$ cannot appear on the right
side of (1). Indeed, if $d(q,r)$ were to appear, then there would exist an
equivalence class of lines $L_t$ such that any line of $L_t$
intersects the segment $[p_i,p_j]$ and $[q,r]$ is an extremal segment
for the couple $(K'_t,K''_t).$ By the definition of an extremal
segment, there exists a line $\ell'$ of $L_t$ passing via $[q,r]$ such that the
convex sets $K'_t$ and $K''_t$ belong to different closed halfplanes
defined by $\ell'$. Since $q, r, p_i, p_j$ are pairwise distinct, the
points $p_i, p_j$ are located in different open halfplanes defined by
$\ell'$. But then $\ell'$ necessarily cuts the segment $[p_i,p_j],$ a
contradiction.

\medskip\noindent
{\bf Case 2:} The points $p_i,p_j,q,r$ are distinct and some line $\ell$ containing $[q,r]$ cuts the segment $[p_i,p_j]$.
\medskip

We assert that in this case there exist precisely four equivalence
classes of lines which separate $p_i$ and $p_j$ in such a way that
$[q,r]$ is an extremal segment; moreover, $[q,r]$ occurs twice as
positive and twice as negative extremal segment. We will identify
these equivalence classes by giving the four pairs $(K'_t,K''_t)$.

Since the line $\ell$ cuts the segment $[p_i,p_j],$ $\ell$ separates
the set $Q-\{ q,r\}$ into the necessarily nonempty sets $A$ and
$B$. Let $\mathring{H}^-_{\ell}, \mathring{H}^+_{\ell}$ denote the open halfplanes containing
respectively the sets $A$ and $B$. Here are the four possible choices
for the pairs $(K'_t,K''_t):$

\medskip\noindent
1. $(\conv(A \cup \{r\}),\conv(B \cup \{q\}))$,

\noindent
2. $(\conv(A \cup \{q\}),\conv(B \cup \{r\}))$,

\noindent
3. $(\conv(A),\conv(B \cup \{q,r\}))$,

\noindent
4. $(\conv(A\cup\{ q,r\}),\conv(B))$.

\medskip\noindent By Lemma \ref{cases2-4}, in each of four choices  the
respective partition $(Q'_t,Q''_t)$ of $Q$ is line-separable. We immediately
conclude that in the first two
pairs $[q,r]$ occurs as a positive extremal segment while in the last
two pairs $[q,r]$ occurs as a negative extremal segment.

\medskip\noindent
{\bf Case 3:} The segments $[p_i,p_j]$ and $[q,r]$ are distinct but $r=p_i$.
\medskip

We assert that in this case there exist precisely two equivalence
classes of lines which separate $p_i$ and $p_j$ in such a way that
$[q,r]$ is an extremal segment; moreover, $[q,r]$ occurs once as
positive and once as negative extremal segment. We will identify these
equivalence classes by giving the two pairs $(K'_t,K''_t)$. We use the
notation of Case 2 except that we assume that $p_j$ is in $B$ and we
allow $A$ to be empty.  Here are the two possible choices for the
pairs $(K'_t,K''_t):$

\medskip\noindent
1. $(\conv(A\cup\{ r\}),\conv(B\cup \{ q\})),$

\noindent
2. $(\conv(A\cup\{ q,r\}),\conv(B)).$

\medskip
Again, by Lemma \ref{cases2-4}, the respective partitions
$(A\cup \{ r\},B\cup \{ q\})$ and $(A\cup\{ q,r\},B)$ of $Q$ are
line-separable. The first pair leads to a
positive extremal segment and the second pair to a negative extremal
segment.
%

\medskip\noindent
{\bf Case 4:} The segments $[p_i,p_j]$ and $[q,r]$ coincide.
\medskip

We assert that in this case there exist exactly two equivalence
classes of lines which separate $p_i$ and $p_j$ in such a way that
$[q,r]$ is an extremal segment. The sets $A$ and $B$ are defined as in
Cases 2 and 3; here we allow that either $A$ or $B$ to be empty.  The
two possible choices for the pairs $(K'_t,K''_t)$ are:

\medskip\noindent
1. $(\conv(A\cup\{ q\}),\conv(B\cup \{ r\})),$

\noindent
2. $(\conv(A\cup\{ r\}),\conv(B\cup \{ q\})).$

\medskip
By Lemma \ref{cases2-4},  the partitions $(A\cup\{
q\},B\cup \{ r\})$ and $(A\cup\{ r\},B\cup \{ q\})$ are
line-separable. Each pair leads to a positive extreme segment, thus to
a contribution of $+d(q,r)=+d(p_i,p_j)$ to the right side of the
equation (1).

The Cases 1-4 show the validity of (1). Since  $\sigma_t\ge 0, t=1,\ldots, m,$
in order to deduce that $(Q,d)$ is $l_1$-embeddable,
it suffices to consider each
$\sigma_t$ with coefficient $\frac{1}{2}.$
\end{proof}

\section{Proof of main results}\label{sec:main}
\subsection{Proof of Theorem~\ref{main}}
Let $(S,d)$ be a Busemann surface.  For a finite
set $Q$ of $S$, denote by $N(Q)$ the number of different collinear
triplets of points of $Q$. Clearly, $N(Q)=0$ if and only if the points
of $Q$ are in general position.

First we will show that for any finite set (not necessarily in
general position) $Q$  of $S,$ the metric space $(Q,d)$ is
$l_1$-embeddable.  
Let $Q=\{ p_1,\ldots, p_n\}$. For a given $\epsilon>0$, in
$m\le n^3$ steps we will define a set of points $Q_{\epsilon}=\{
p'_1,\ldots, p'_n\}$ in general position such that
$d(p_i,p'_i)<\epsilon/2$ and $|d(p_i,p_j)-d(p'_i,p'_j)|<\epsilon$ for
any $p_i,p_j\in Q$.  For this, setting $Q_0:=Q,$ we will construct a
sequence of sets $Q_1,\ldots,Q_m$ such that $N(Q)=N(Q_0)>N(Q_1)>\ldots
>N(Q_{m-1})>N(Q_m)=0$. Each set $Q_{i+1}$ is obtained from the set
$Q_i$ by moving a single point of $Q_i$ at distance
$<\frac{\epsilon}{2m}$. We will set $Q_{\epsilon}:=Q_m$ and denote by
$p'_i$ the final position of the point $p_i$ after all these
movements. Since each initial point can move at most $m$ times,
$d(p_i,p'_i)\le \epsilon /2,$ thus for each pair $p_i,p_j\in Q$ we
will have $|d(p_i,p_j)-d(p'_i,p'_j)|<\epsilon$.

We will describe now how from a set $Q$ with $N(Q)>0$ to define a new
set $Q_1$ with $N(Q_1)<N(Q).$ Let $p,q,r$ be three points of $Q$ such
that $q\in [p,r]$.  Let $\mathcal R$ denote the set of pairs $\{
p',q'\}$ of points of $Q$ such that $q,p',q'$ are not collinear. This
means that for any pair $\{ p',q'\}\in {\mathcal R},$ the point $q$
does not belong to the geodesic $[p',q']$ and to the cones
$C(p',q')=\{ x\in S: p'\in [x,q']\}$ and $C(q',p')=\{ x\in S: q'\in
[p',x]\}$. Since the set $\mathcal R$ is finite, the sets
$R(p',q'):=C(p',q')\cup [p',q']\cup C(q',p')$ are closed and $q$ does
not belong to any such set, we conclude that $\epsilon':=\min\{
d(q,R(p',q')): \{ p',q'\}\in {\mathcal R}\}>0.$ Let $\epsilon_0=\min
\{\epsilon', \frac{\epsilon}{2m}\}.$ Since $q$ contains a neighborhood
homeomorphic to a circle, there exists a direction in
the neighborhood of $q$ different from the directions on $[q,p]$ and
$[q,r].$ Let $q_0$ be a point obtained from $q$ by moving along that
direction at distance $<\epsilon_0$ from $q$. Let $Q_1:=Q\setminus \{
q\} \cup \{ q_0\}$.  We assert that $N(Q_1)<N(Q).$ From the
construction, $q_0\notin R(p',q')$ for any pair $\{ p',q'\}\in
{\mathcal R},$ thus $q_0$ cannot create new collinear triplets. On the
other hand, since $q_0\notin [p,q]$ we conclude that indeed
$N(Q)>N(Q_1).$

As a result, for each $\epsilon > 0$ we can define a set
$Q_{\epsilon} = \{p'_1,\ldots, p'_n\}$ of points in general position
such that $d(p_i,p'_i) < \epsilon$ and
$|d(p_i,p_j)-d(p'_i,p'_j)| < \epsilon$. By Proposition \ref{crofton},
the metric spaces $(Q_{\epsilon},d)$ are $l_1$-embeddable. On the set
$Q$ we define the metric $d_{\epsilon}$ by setting
$d_{\epsilon}(p_i,p_j) = d(p'_i,p'_j)$ for any two points $p_i,p_j \in
Q$, where $p'_i$ and $p'_j$ are the images of $p_i$ and $p_j$ in the
set $Q_{\epsilon}$. Since $(Q_{\epsilon},d)$ are $l_1$-embeddable, the
metric spaces $(Q,d_{\epsilon})$ are also $l_1$-embeddable, whence
each $d_{\epsilon}$ belongs to the cut cone CUT$_n$. Since CUT$_n$ is
closed and $d_{\epsilon}$ converge to $d$ when $\epsilon$ converges to
$0$, we conclude that $d \in \mbox{CUT}_n$. This establishes that the
finite metric space $(Q,d)$ is $l_1$-embeddable. Since all finite
subspaces of $(S,d)$ are $l_1$-embeddable, the metric space $(S,d)$ is
$L_1$-embeddable by the compactness result of \cite{BrDhCaKr}.

\subsection{Proof of Corollary \ref{graph}}
First we  present an example of a
Busemann graph $B_n$, which is not isometrically embeddable into
$l_1$. The graph $B_n$ consists of two 3-cycles $T'=u'uu''$ and
$T''=v'vv''$ and two odd $(2n+1)$-cycles $C',C''$ sharing the edge
$uv$.  Let $C' \cap T' = u'u, C' \cap T'' = v'v$ and $C'' \cap T' =
u''u, C'' \cap T''= v''v$; see
Fig. \ref{fig-planar-no-isometric-embedding}.  Let $X_n$ be the planar
polygonal complex obtained by replacing the cycles $C',C''$ by regular
$(2n+1)$-gons and the 3-cycles $T',T''$ by equilateral triangles. Now,
if $n\ge 6,$ then the angles around the vertices $u$ and $v$ are $\ge
2\pi$. Hence $X_n$ is a CAT(0) complex and $B_n$, its 1-skeleton, is a
Busemann graph. To show that a graph $G$ is not $l_1$-embeddable, we will
use the well-known fact (which can be derived from the pentagonal
inequality for $L_1$-spaces, see \cite{DeLa}) that all intervals
$I(a,b)$ in $l_1$-graphs are convex (i.e., if $c',c''\in I(a,b)$ and
$c\in I(c',c'')$, then $c\in I(a,b)$; the interval $I(a,b)$ consists
of all vertices on shortest $(a,b)$-paths between $a$ and $b$).  Let
$x$ and $y$ be the vertices of $C'$ and $C''$ opposite to the edge
$uv$. Then $u',u'',v',v''\in I(x,y),$ however $u,v\in
I(u',v')\setminus I(x,y),$ thus $I(x,y)$ is not convex and therefore
$B_n$ is not an $l_1$-graph.

\begin{figure}[h]
\begin{center}
\begin{tikzpicture}[x=2cm,y=2cm]
\foreach \i [evaluate=\i as \angle using (2 * \i + 1) * 360 / 26] in {0,1,2,...,12}
 { \draw (\angle:1) node[circle,inner sep = 0pt,minimum size =3pt,fill = black] {};
   \draw ($2*cos(360/26)*(1,0) - (\angle:1)$) node[circle,inner sep = 0pt,minimum size =3pt,fill = black] {};
   \draw (\angle:1) -- (\angle + 360/13:1);
   \draw ($2*cos(360/26)*(1,0) - (\angle:1)$) -- ($2*cos(360/26)*(1,0) - (\angle + 360/13:1)$);
 }
\draw (3*360/26:1) -- ($2*cos(360/26)*(1,0) - (23*360/26:1)$);
\draw (23*360/26:1) -- ($2*cos(360/26)*(1,0) - (3*360/26:1)$);
\draw (360/26:1) node[anchor = west] {$u$};
\draw (3*360/26:1) node[anchor = south] {$u'$};
\draw (13*360/26:1) node[anchor = east] {$x$};
\draw (23*360/26:1) node[anchor = north] {$v'$};
\draw (25*360/26:1) node[anchor = west] {$v$};
\draw ($2*cos(360/26)*(1,0) - (23*360/26:1)$) node[anchor = south] {$u''$};
\draw ($2*cos(360/26)*(1,0) - (13*360/26:1)$) node[anchor = west] {$y$};
\draw ($2*cos(360/26)*(1,0) - (3*360/26:1)$) node[anchor = north] {$v''$};
\end{tikzpicture}
\end{center}
\caption{A Busemann graph that admits no isometric embedding into $l_1$.}
\label{fig-planar-no-isometric-embedding}
\end{figure}
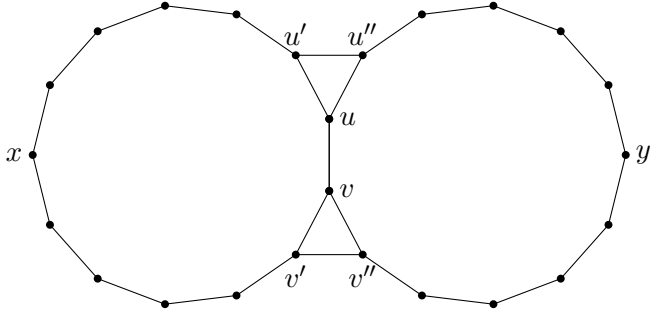

Now, we prove Corollary \ref{graph} that for any Busemann graph
$G=(V,E)$ with a standard graph metric $d_G$, $(V,d_G)$ admits an
embedding into $L_1$ with constant distortion. Every (combinatorial)
graph $G=(V,E)$ equipped with its standard distance $d_G$ can be
transformed into a (network-like) geodesic space $G'=(V',d_{G'})$ by
replacing every edge $e=(u,v)$ by a segment $\gamma_{uv}=[u,v]$ of
length 1; the segments may intersect only at common ends.  Then
$(V,d_G)$ is isometrically embedded in a natural way in $(V',d_{G'}).$
$G'$ is often called a {\it metric graph}. Since $G$ is a Busemann
graph, $G$ is the 1-skeleton of a non-positively curved regular planar
complex $X$; let $d$ be the intrinsic $\ell_2$-metric on $X$. The
graph $G$ and its metric graph $G'$ are naturally embedded in $X$ by
the identity mapping $\id: V' \rightarrow X$.

As we noticed in Subsection 2.4, $(X,d)$ can be extended to a Busemann
surface, and this extension is an isometric embedding. Thus Theorem
\ref{main} implies that $(X,d)$ is isometrically embeddable into
$L_1$. Therefore, in order to show that $(V,d_G)$ is embeddable into
$L_1$ with distortion $c = 2+ \pi / 2$, it suffices to show that $\id$ embeds
$(V',d_{G'})$ into $(X,d)$ with
distortion $c$. Pick any two points $x,y \in V'$ and let $[x,y]$ be the
geodesic segment between $x$ and $y$ in $X$. Let
$x =: x_0,x_1,\ldots, x_{k-1},x_k := y$ be the consecutive intersections of
$[x,y]$ with the 1-faces (edges) of $X$. Then each pair of consecutive
points $x_{i-1},x_i$ belongs to a common 2-face $F_i$ of $X$. Let
$P_i$ be the shortest of the two boundary paths of $F_i$ connecting
$x_{i-1}$ and $x_i$ and let $\ell_i$ be its length. Since the union
$\cup_{i=1}^k P_i$ of these paths is a path between $x$ and $y$ in the
metric graph $G',$ we deduce that $d_{G'}(x,y) \le \sum_{i=1}^k
\ell_i$. Therefore, to prove that $d_{G'}(x,y) \le c \cdot d(x,y)$ it
suffices to show that ${\ell}_i \le c \cdot d(x_{i-1},x_i)$ for all
$i=1,\ldots, k$, where $d(x_{i-1},x_i)$ is simply the Euclidean
distance between two boundary points $x_{i-1},x_i$ of the regular
polygon $F_i$. This is a consequence of the following result:

\begin{lemma}\label{cl:shortcut-face}
If $a,b$ are two points on the boundary of a regular Euclidean polygon
$F$ and $\ell(a,b)$ is the length of the shortest boundary path $P$ of
$F$ connecting $a$ and $b,$ then $\ell(a,b) \leq (2+\pi/2)d(x,y)$.
\end{lemma}

\begin{proof} Let $C$ be the circle circumscribed to the regular polygon $F$. We distinguish two cases.

\medskip\noindent{\bf Case 1.} $a$ and $b$ lie on incident edges $[z',z], [z,z'']$ of $F$.
Let $\alpha$ be the angle between these edges.  Then $d(a,b)
= (d(a,z)^2 + d(b,z)^2 - 2d(a,z)d(b,z) \cos \alpha)^{\frac{1}{2}}$. Simple
calculations using the fact that $\pi/3\le \alpha\le \pi$
show that if $d(a,z) + d(z,b)$ is fixed, then $d(a,b)$ is
minimized when $\alpha$ is minimized (i.e., $\alpha=\pi/3$) and $d(a,z) =
d(z,b)$, in which case $d(a,b) \geq (d(a,z) + d(z,b))/2\ge \ell(a,b)/2$.\medskip

\noindent{\bf Case 2.} $a$ and $b$ lie on  non-incident edges, say
$P=(a,z_1,z_2,\ldots z_{k},b)$, with $k \geq 2$. Then clearly $d(a,b)
\geq 1$ and $d(a,b) \geq d(z_1,z_k)$. Also $\ell(z_1,z_k)$ (the length of the portion of
$P$ between  $z_1$ and $z_k$) is upper bounded by the length of the
$(z_1,z_k)$-arc of the circle $C$, and thus is at most $\frac{\pi}{2} d(z_1,z_k)$.
From this we get
$$\ell(a,b)\leq 2 + \ell(z_1,z_k) \leq 2d(a,b) + \frac{\pi}{2} d(z_1,z_k) \leq \left(2 + \frac{\pi}{2}\right) d(a,b).$$
\end{proof}

Hence, $\id: V'\rightarrow X$ is a non-expansive embedding  of  $G'$ (and $G$) into $X$ with distortion
$c = 2 + \pi / 2$: for any two
vertices $x,y$ of $G$ (and, more generally, for any two points $x,y \in
V'$ of $G'$) we have $\frac{1}{c}d_{G'}(x,y) \le d(x,y) \le
d_{G'}(x,y)$. This concludes the proof of Corollary \ref{graph}.


\end{document}